\documentclass[a4paper,12pt]{amsart}
\usepackage{graphicx} % Required for inserting images
\usepackage{amsmath,amssymb,amsthm,url,stmaryrd,mathrsfs,txfonts,bm}
\usepackage{xcolor}
\numberwithin{equation}{section}

\newcommand{\ff}{\mathbb{F}}

\newcommand{\Cc}{\mathscr{C}}

\newcommand{\cb}{\mathbf{c}}
\newcommand{\xb}{\mathbf{x}}
\newcommand{\yb}{\mathbf{y}}
\newcommand{\vb}{\mathbf{v}}
\newcommand{\wt}{\mathrm{wt}}

\newtheorem{lem}{Lemma}
\newtheorem{cor}{Corollary}
\newtheorem{prop}{Proposition}
\newtheorem{thm}{Theorem}
\theoremstyle{definition}
\newtheorem{defn}{Definition}
\newtheorem{ex}{Example}
\newtheorem{rem}{Remark}

\newcommand{\ds}{\displaystyle}

\newcommand{\M}{{\rm M}}

\newcommand{\ab}{\boldsymbol{\alpha}}

\begin{document}
\title{On the existence of MRD self-dual codes}
\author{Gr\'{e}gory Berhuy }

\address{Institut Fourier, Universit\'{e} Grenoble Alpes, France}
\email{gregory.berhuy@univ-grenoble-alpes.fr}
\date{\today}

\maketitle

\begin{abstract}
 In this paper, we investigate the existence of self-dual MRD codes $C\subset L^n$, where $L/F$ is an arbitrary field extension of degree $m\geq n$. We give necessary conditions on $L/F$  to have the existence of such codes, as well as sufficient conditions on $F$ to have their non-existence.
 We also give a full answer when $L/F$ is cyclic and $m=n \equiv 2 \ [4]$, and apply our results to finite fields. Along the way, we also prove the non-existence of self-dual MRD codes for another type of duality.
\end{abstract}



\tableofcontents

\section{Basic definitions and summary of the main results}

Let $F$ be an arbitrary field. Let $m,n\geq 1$ be two integers such that $m\geq n$. In \cite{Del}, Delsarte introduced the notion of rank metric (linear) code: a {\it Delsarte rank metric code} is an $F$-linear subspace $\Cc$ of the space $\M_{m\times n}(F)$ of $m\times n$ matrices with entries in $F.$

The {\it rank distance} $d_1(\Cc)$ of $\Cc$ is the integer $$d_1(\Cc)=\min_{M\in\Cc\setminus\{0\}}({\rm rk}(M)). $$

It is well-known (see \cite[Thm 5.4]{Del}, for example) that we have the {\it Singleton bound} $$d_1(\Cc)\leq n-\dfrac{\dim_F(\Cc)}{m}+1.$$

The rank metric code $\Cc$ is said to be {\it MRD} if $d_1(\Cc)= n-\dfrac{\dim_F(\Cc)}{m}+1.$  

In particular, the dimension of an MRD code is necessarily a multiple of $m.$

The {\it dual code} of $\Cc$ is the orthogonal space $\Cc^\perp$, where the orthogonal is defined with respect to the standard inner product on $\M_{m\times n}(F)$, that is, the $F$-bilinear form $$(M,N)\in\M_{m\times n}(F)\times \M_{m\times n}(F)\mapsto {\rm tr}(MN^t)\in F. $$

Later on, Gabidulin proposed in \cite{Gab} another definition of a rank metric code, as follows.
Let $n\geq 1$ be an integer, and let $L/F$ be a field extension of finite degree $m\geq n$. 

An $L$-linear subspace $C\subset L^n$ then gives rise naturally to a Delsarte rank metric code $M_{\ab}(C)\subset \M_{m\times n}(F),$ as we explain now.

Let $\ab=(\alpha_1,\ldots,\alpha_m)$ be an $F$-basis of $L$.

If $\cb=(c_1,\ldots,c_n)\in L^n$, write $c_j=\ds\sum_{i=1}^m c_{ij}\alpha_i, c_{ij}\in F$, and set $M_{\ab}({\bf c})=(c_{ij})\in\M_{m\times n}(F).$

Then $M_{\ab}(C)=\{M_{\ab}(\cb) \mid \cb\in C\} $ is an $F$-linear subspace of $\M_{m\times n}(F)$ of dimension $m\dim_L(C).$

The matrix $M_{\ab} (\cb)$ obviously depends on the choice of the $F$-basis $\ab$ of $L$. However, its rank does not, since another choice of basis multiply $M_{\ab}(\cb)$ on the left by an invertible $m\times m$ matrix with entries in $F$.

We define {\it the weight} of $\cb$, denoted by $\wt(\cb)$, to be the rank of $M_{\ab}(\cb)$, where $M_{\ab} (\cb)$ is computed with respect to any choice of an $F$-basis $\ab$ of $L.$

The {\it rank distance} of $C$ is then defined as $$d_1(C)=d_1(M_{\ab}(C))=\min_{\cb\in C\setminus\{0\}}(\wt(\cb)).$$

The code $C\subset L^n$ is said to be MRD if $M_{\ab}(C)$ is.

In this context, the Singleton bound may be rewritten as $$d_1(C)\leq n-\dim_L(C)+1, $$
so $C$ is MRD if and only if $d_1(C)=n-\dim_L(C)+1.$

An $L$-linear subspace $C$ of $L^n$ will be called a {\it Gabidulin rank metric code}. Such a code may then be viewed as a special case of a Delsarte rank metric code. By definition, the two notions of rank distances agree.

We now define the {\it dual code} of $C$ to be the orthogonal space $C^\perp,$ where the orthogonal is defined with respect to the standard inner product on $L^n,$ that is $({\bf x}, {\bf y})\in L^n \times L^n\mapsto {\bf x}{\bf y}^t\in L.$

Note that both notions of rank metric codes and dual codes have been originally defined over finite fields but readily extend to arbitrary fields.

A natural question is:

does there exist a Gabidulin or Delsarte self-dual MRD code ?

We should insist on the fact that in general, there is no obvious relation between the two previous notions of self-duality. More precisely, if $C\subset L^n$ is self-dual for the standard inner product of $L^n$, the corresponding rank metric code $M_{\ab}(C)$ is not necessarily self-dual anymore for the form $$(M,N)\in\M_{m\times n}(F)\times \M_{m\times n}(F)\mapsto {\rm tr}(MN^t). $$ However, as already pointed out by Ravagnani in \cite{Rav} in the case of finite fields, both notions of duality will agree for a suitable choice of the $F$-basis $\ab$ of $L.$ We will address this delicate question in the case of arbitrary fields in this paper, generalizing some results of Ravagnani along the way.

Note that the dimension over $F$ of a self-dual code is necessarily equal to $\dfrac{mn}{2}$, so $mn$ has to be even.

In \cite{Neb}, Nebe and Willems investigate the existence of Delsarte self-dual MRD codes in the case $n=m$ (so $n$ is necessarily even), and  prove the following results (among other things):

    \begin{enumerate}
        \item If ${\rm char}(F)=2,$ there is no self-dual MRD code $\Cc\subset\M_n(F).$

\medskip

\item Assume that $F=\ff_q$, where $q$ is odd.

\medskip

\begin{enumerate}

        \item If $n=2$, self-dual MRD codes exist if and only if $-1$ is not a square in $F.$

        \item If $n\equiv 2 \ [4]$ and $-1$ is not a square in $F,$ there exist self-dual MRD codes.
\end{enumerate}
        
    \end{enumerate}

In the same paper, the authors notice also that there is no known example of self-dual MRD codes $\Cc \subset\M_n(\ff_q)$ when $-1$ is a square in $\ff_q$ or $n\equiv 0 \ [4].$

In this paper, we will prove some existence and non-existence results of Gabidulin self-dual MRD codes $C\subset L^n$, where $F$ is an arbitrary field and $L/F$ is a finite degree field extension.

{\bf Notational convention. }In this paper, Gabidulin rank metric codes will be denoted by the letter $C$, and Delsarte rank metric codes will be denoted by the letter $\mathscr{C}.$

We now describe the structure of this paper. Unless specified otherwise, the codes we consider now are Gabidulin rank metric codes.

After proving two preliminary results on the rank distance in Section 2, we prove in Section 3 the non-existence of Gabidulin MRD self-dual codes for another type of duality (namely,  the duality with respect to the standard hyperbolic form). 

 In Section 4, we prove that there are no MRD self-dual codes if $F$ has characteristic two. In the case where $F$ has odd characteristic, we then use the results of Section 3 to translate the existence or non-existence properties of self-dual MRD codes in terms of representations of $-1$ as a sum of squares.  
We will then apply our results to the case of finite fields. In particular, if $F$ is a finite field,  we get necessary and sufficient conditions for the existence a self-dual MRD code in the case where $m=n.$

In Section 5, we will relate the two notions of duality introduced previously, generalizing the work of Ravagnini. In particular, in the case of finite fields, we will prove that if $C\subset L^n$ is self-dual, then $M_{\ab}(C)$ is also self-dual for a suitable choice of $\ab$.

We also put our results in perspective with the known existence results for Delsarte rank metric codes.

Finally, in Section 6, we go back to the characteristic two case, and we propose a suitable substitute for MRD self-dual codes.

\section{Preliminaries on the rank distance}

In this short section, we will prove some useful lemmas on the rank distance.

Let $F$ be an arbitrary field, and let $L/F$ be a field extension of degree $m$. 

Finally, let $n\geq 1$ be an integer such that $m\geq n$.

We will fix once and for all an $F$-basis $\ab=(\alpha_1,\ldots,\alpha_m)$ of $L.$

We will need the following basic result of linear algebra.

\begin{lem}\label{mc}
For any $\cb\in L^n,$ $M_{\ab}(\cb)$ is the unique matrix $M\in \M_{m\times n}(F)$ satisfying $\cb=\ab M.$
\end{lem}

%\begin{proof}
%If $M=(m_{ij})\in \M_{m\times n}(F)$ is an arbitrary matrix, we have 
%$$ \ab M=(\sum_{i=1}^m m_{i1}\alpha_i,\ldots,\sum_{i=1}^m m_{in}\alpha_i).$$
%
% The equality   $\cb=\ab M_{\ab}(\cb)$ then comes from the definition of $M_{\ab}(\cb).$
%
% To prove the uniqueness part, it is enough to prove that, if $M\in \M_{m\times n}(F)$ satisfies $\ab M=0$, then $M=0$. But this comes from the equality above and the fact that $\alpha_1,\ldots,\alpha_m$ are $F$-linearly independent.
%\end{proof}

The following easy lemma will be crucial in the sequel.

\begin{lem}\label{fl}
 Let $f:L^n\overset{\sim}{\to} L^n$ be an $L$-automorphism of $L^n$. Let $P\in \M_n(L)$ be the matrix representation of $f$ with respect to the canonical basis of $L^n$, that is, the unique matrix $P$ such that $f({\bf c})={\bf c}P^t$ for all ${\bf c}\in L^n$.
 
Assume that $P\in {\rm GL}_n(F).$ Then, for all $\cb\in L^n$, we have $\wt(f(\cb))=\wt(\cb).$

 In particular, for every linear code $C\subset L^n,$ we have $d_1(f(C))=d_1(C).$
\end{lem}

\begin{proof}
If $\cb=(c_1,\ldots,c_n)\in L^n$, we have $f(\cb)=\cb P^t$, and thus $M_{\ab}(f(\cb))=\ab  M_{\ab}(\cb) P^t$ by Lemma \ref{mc}. The uniqueness part of this lemma then gives us $M_{\ab}(f(\cb))=M_{\ab}(\cb)P^t.$

Since $P$ is invertible, the lemma follows.
\end{proof}

The next lemma is certainly well-known. Since it will be very useful to establish our main results, we prove it for sake of completeness.

\begin{lem}\label{wt}
For any $\cb=(c_1,\ldots,c_n)\in L^n,$ we have 
$$\wt(\cb)=\dim_F({\rm Span}_F(c_1,\ldots,c_n)).$$   
\end{lem}

\begin{proof}
Let $\cb=(c_1,\ldots,c_n)\in L^n.$ If $\ab=(\alpha_1,\ldots,\alpha_m)$ is an $F$-basis of $L$, 
the map $\ds\sum_{i=1}^m x_i\alpha_i\mapsto (x_1,\ldots,x_m)\in F^m$ is an $F$-linear isomorphism which maps ${\rm Span}_F(c_1,\ldots,c_n)$ onto the column space of $M_{\ab}(\cb),$ hence the result. 
\end{proof}

\section{Lagrangian codes}

In this section and the next ones, we will make use of the following notation.

{\bf Notation. }If $n=2d$ is an even integer and $K$ is a field, a matrix $M\in\M_{k\times n}(K)$ will be sometimes denoted by $M=(M_1\mid M_2),$ where $M_i\in\M_{k\times d}(K),$
and a vector ${\bf c}\in K^n$ will be written as ${\bf c}=({\bf v}\mid {\bf w}),$ where ${\bf v},{\bf w}\in K^d.$

Let $F$ be an arbitrary field, and let $n=2d\geq 2$ be an even integer. Set $H_n=\begin{pmatrix}
0 & I_d \cr I_d & 0
\end{pmatrix}\in\M_n(L),$ where $I_d$ is the identity matrix of size $d$.

We define the canonical hyperbolic on $L^n$  $h_{n,L}: L^n\times L^n\to L$ by $$h_{n,L}({\bf x},{\bf y})={\bf x}H_n{\bf y}^t \ \mbox{ for all }{\bf x},{\bf y}\in L^n. $$

We will denote $C^{\perp_n}$ the orthogonal of $C$ with respect to $h_{n,L}$, that is $$C^{\perp_n}=\{ \xb\in L^n\mid h_{n,L}(\xb,\cb)=0 \ \mbox{ for all } \cb\in C\}.$$

\begin{rem}\label{remcn}
If $C$ is defined as the row space of a matrix $G=(G_1\mid G_2),$ where $G_i\in\M_{k\times d}(L)$ for some integer $k\geq 1$, then easy block matrix computations show that $C^{\perp_n}=\ker(G_2\mid G_1)^t$.

In particular, we will have $C\subset C^{\perp_n}$ if and only if $G_2G_1^t+G_1G_2^t=0$, that is, if and only if $G_1G_2^t$ is a skew-symmetric matrix.
\end{rem}

In this section, we will study linear codes which are self-dual with respect to $h_{n,L}$.
Note that, since $h_{n,L}$ is non-degenerate, we have $$\dim_L(C^{\perp_n})=n-\dim_L(C).$$ In particular, if $C=C^{\perp_n}$, we get $\dim_L(C)=d.$

Since self-dual spaces are called Lagrangian spaces in quadratic form theory, we introduce the following terminology.

\begin{defn}
Let $n=2d$ be an even integer.

 A  {\it Lagrangian code} of $L^n$ is a subspace $C$ of $L^n$ such that $C=C^{\perp_n}.$ Equivalently, a code $C\subset L^n$ is a Lagrangian code if $\dim_L(C)=d$ and $C\subset C^{\perp_n}$.

In particular, if $C$ is  the row space of a full rank matrix $G=(G_1\mid G_2),$ where $G_i\in\M_{k\times d}(L)$ for some integer $k\geq 1$, then $C$ is a Lagrangian code if and only if $k=d$ and $G_1G_2^t$ is a skew-symmetric matrix.
\end{defn}

We now come to the main result of this section.

\begin{thm}\label{MRDhyp}
Let $n=2d$ be an even integer, and let $L/F$ be a field extension of degree $m\geq n,$ where $F$ has odd characteristic. For any Lagrangian code $C$, there exists ${\bf c}\in C\setminus\{0\}$ such that ${\rm wt}(\cb)\leq d$.

In particular, any Lagrangian code satisfies $d_1(C)\leq d.$ In other words, there are no MRD Lagrangian codes.
\end{thm}

\begin{proof}
Let $C_0=\{ (x_1,\ldots,x_d,0,\ldots,0) \mid x_1,\ldots,x_d\in L\}$. Note that, since any element of $C_0$ has at most $d$ non-zero coordinates,  
we have ${\rm wt}(\cb_0)\leq d$ for all $\cb_0\in C_0$ by Lemma \ref{wt}.

Now let $C$ be an arbitrary Lagrangian code. Assume by way of contradiction that, for all ${\bf c}\in C\setminus\{0\}$, we have ${\rm wt}(\cb)\geq d+1$. The previous observation then shows that $C\cap C_0=\{0\}$. This implies that the $L$-linear map $$\cb=(c_1,\ldots,c_n)\in C\mapsto (c_{d+1},\ldots,c_n)\in L^d$$ is injective.
Since $C$ is a Lagrangian code, we have $\dim_L(C)=d=\dim_L(L^d)$, and the map above is therefore an isomorphism. Taking the preimages of the elements of the canonical basis of $L^d$ then yields the existence of a matrix $S\in\M_d(L)$ such that $C$ is the row space of $G=(S\mid I_d)$.
Since $C\subset C^{\perp_n}$, $S$ is a skew-symmetric matrix by Remark \ref{remcn}.

Since $F$ has odd characteristic (and thus, so has $L$), it follows that the diagonal entries of $S$ are zero. Hence, any row of $G$ is a non-zero element of $C\setminus\{0\}$ whose weight is at most $d$, leading to the desired contradiction.
This concludes the proof.
\end{proof}

Note that the very last argument fails in characteristic two, since a skew-symmetric matrix may have non-zero diagonal coefficients in this case. 
In fact, we will see in an upcoming section of this paper that Lagrangian MRD codes may exist in characteristic two.

\section{Self-dual MRD codes}

We now come back to the question of the existence of MRD self-dual codes (with respect to the canonical inner product of $L^n$).

We start with the easy case of fields of characteristic two (compare with \cite[Thm 3.1]{Neb}).

\begin{prop}\label{car2}
 Let $F$ be a field of characteristic two. Let $L/F$ be a field extension of degree $m,$ and let $C\subset L^n$ be a linear code.

 If $C\subset C^\perp$, then $d_1(C^\perp)\leq 1.$

 In particular, there are no MRD self-dual codes in $L^n$.
\end{prop}

\begin{proof}
Assume that $C\subset C^\perp.$ Then, for all $\cb=(c_1,\ldots,c_n)\in C$, we have $\cb \cb^t=0=c_1^2+\cdots+c_n^2=(c_1+\cdots+ c_n)^2.$

Hence $\cb \vb^t=0$ for all $\cb\in C,$ where $\vb=(1,\ldots,1)$. In other words, $\vb\in C^\perp.$ By Lemma \ref{wt}, ${\rm wt}(\vb)=1$, and the result follows.
\end{proof}

Throughout the rest of this section, we assume that $F$ has odd characteristic, unless specified otherwise.

We first derive a necessary condition to have a self-dual code.

Note that the equality $C=C^\perp$ implies as usual that $n=2\dim_L(C)$, so $n$ is necessarily even. A self-dual code will then have dimension $\dfrac{n}{2}$, and will be MRD if and only if $d_1(C)=\dfrac{n}{2}+1.$

We will then assume once and for all that $n=2d,$ where $d\geq 1.$

\begin{lem}\label{selfdualexist}
Let $L/F$ be a field extension of degree $m$, and let $n$ be an even integer.  Assume that there exists a self-dual code $C\subset L^n.$
Then the standard inner product of $L^n$ is isomorphic to the hyperbolic form $h_{n,L}$ over $L.$ 
\end{lem}

\begin{proof}
    Assume that $C\subset L^n$ is a self-dual code. In other words, $C$ is a totally isotropic subspace of dimension $d$ with respect to the standard inner product of $L^n$. By \cite[Chapter 1,Theorem 4.5 (iv)]{Sch}, this implies that this form is isomorphic to $h_{n,L}$ over $L$.
\end{proof}

The necessary condition above is not very explicit and does not seem very easy to check in practice. We will present later a more evocative way to state it.

For the moment, note that the conlusion of Lemma \ref{selfdualexist} obviously holds if it is already holds over $F$, that is, if the standard inner product of $F^n$ is isomorphic over $F$ to the standard hyperbolic form $h_{n,F}$.

Unfortunately, in this case, we have the following negative result.

\begin{thm}\label{nonexist}
Let $L/F$ be a field extension of degree $m$, and let $n=2d$ be an even integer, with $m\geq n$.  Assume that the standard inner product of $F^n$ is isomorphic over $F$ to the standard hyperbolic form $h_{n,F}.$

Then, there are no MRD self-dual codes in $L^n.$ 
\end{thm}

\begin{proof}
Let $f:F^n\overset{\sim}{\to}F^n$ be an isomorphism between the standard inner product of $F^n$ and $h_{n,F}$. In other words, we have $\xb\yb^t=h_{n,F}(f(\xb),f(\yb))$ for all $\xb,\yb\in F^n$.

Let us denote by $f_L:L^n\to L^n$ the canonical extension of $f$ to $L^n$, that is, the unique endomorphism $f_L$ whose matrix representation with respect to the canonical basis of $L^n$ is exactly the matrix representation of $f$ with respect to the canonical basis of $F^n.$
Then $f_L$ is an isomorphism over $L$ between the standard inner product of $L^n$ and $h_{n,L}.$
Hence, if $C\subset L^n$ is self-dual, then $f_L(C)$ is self-dual with respect to $h_{n,L}$. By Theorem \ref{MRDhyp}, $d_1(f_L(C))\leq d$. 
But Lemma \ref{fl} implies that $d_1(C)=d_1(f_L(C))\leq d$, and $C$ is not MRD.
\end{proof}

\begin{cor}\label{corononexist}
Let $L/F$ be a field extension of degree $m$, and let $n=2d$ be an even integer, with $m\geq n$.  If $m$ is odd, there are no MRD self-dual codes in $L^n.$
\end{cor}

\begin{proof}
Assume to the contrary that there is a self-dual MRD code in $L^n$. By Lemma \ref{selfdualexist}, the standard inner product of $L^n$ is isomorphic to $h_{n,L}$. But these two forms are nothing but the extensions to $L$ of the standard inner product of $F^n$ and $h_{n,F}$ respectively.  By Springer's theorem, these latter forms are then isomorphic over $F$ (see \cite[Chapter 2, Corollary 5.4]{Sch}, for example). The previous theorem then yields a contradiction.
\end{proof}

We would like now to rephrase in a more evocative way the assumptions of the previous results. 

Let us recall first some definitions and notation from quadratic form theory. Let $K$ be a field with odd characteristic.

A symmetric bilinear space $(V,b)$ over $K$ is said to be {\it hyperbolic} if it is isomorphic to $(K^n,h_{n,K})$ for some even integer $n\geq 2$.

We say that $(V,b)$ {\it represents} $\lambda\in K^\times$ if there exists $v\in V$ such that $b(v,v)=\lambda$.

We say that $(V,b)$ is {\it multiplicative} if it is either  hyperbolic, or if it satisfies the two following conditions:

(a) $(V,b)$ is anisotropic, that is, for all $v\in V\setminus\{0\}, b(v,v)\neq 0$ 

(b) for all $\lambda\in K^\times,$ $b$ represents $\lambda$ if and only if $b\simeq \lambda b$.

Finally, recall that two non-degenerate symmetric bilinear spaces $(V_1,b_1),(V_2,b_2)$ are {\it Witt equivalent} if $(V_1,b_1)\perp(V_2,-b_2)$ is hyperbolic.

One may show that the set $W(K)$ of Witt equivalence classes of non-degenerate symmetric bilinear spaces is an abelian group for the composition law induced by the orthogonal sum, called the {\it Witt group} of $K$.

We may now state  the next result.

\begin{lem}\label{level}
Let $n$ be an even integer, and let $K$ be a field of odd characteristic. Write $n=2^{s+1}r$, where $s\geq 0$ and $r$ is odd. Then, the following properties are equivalent:

\begin{enumerate}
    \item[(i)] the standard inner product of $K^n$ is isomorphic over $K$ to the hyperbolic form $h_{n,K}$
    
    \smallskip
    
    \item[(ii)] the standard inner product of $K^{2^{s+1}}$
    is isomorphic over $K$ to $h_{2^{s+1},K}$

    \smallskip
    \item[(iii)] $-1$ is a sum of $2^s$ squares in $K.$
\end{enumerate}
\end{lem}

\begin{proof}
If $t\geq 1$ is an integer and $b$ is a symmetric bilinear form, recall that $t\times b$ denotes the sum of $t$ orthogonal copies of $b$, that is $b\perp  \cdots \perp b$. In particular, if $\langle 1\rangle$ denotes the form $(x,y)\in K^2\mapsto xy\in K$, $t\times \langle 1\rangle$ is the standard inner product of $K^t$.

Then, property (i) is equivalent to say that 
$n\times \langle 1\rangle=0$ in the Witt group $W(K)$, that is
$r\times (2^{s+1}\times \langle 1\rangle)=0\in W(K).$

Similarly, (ii) is equivalent to $2^{s+1}\times\langle 1\rangle =0\in W(K).$ Since $r$ is odd, the equivalence of (i) and (ii) comes from the fact that $W(K)$ does not contain any  non-zero element of odd order by \cite[Chapter 2, Thm 10.12]{Sch}.

We now prove that (ii) and (iii) are equivalent.

Assuming we have (ii), we get $2^s\times \langle 1\rangle=-(2^s\times \langle 1\rangle )\in W(K).$ Since these two forms have same rank, this is equivalent to $2^s\times\langle 1\rangle\simeq -(2^s\times \langle 1\rangle)=2^s\times \langle -1 \rangle.$ 

Now the bilinear form on the right represents $-1$. Since isomorphic bilinear forms represent the same values, the bilinear form on the left also represents $-1$. But this form represents the sums of $2^s$ squares, and we get (iii).

Finally, assume that we have (iii).
First, the form $2^s\times \langle 1\rangle$ is a multiplicative form (start with the form $\langle 1\rangle$ and apply \cite[Chapter 2, Lemma 10.4]{Sch} several times with $\alpha=1$).

If $2^s\times \langle 1\rangle$ is hyperbolic, that is, isomorphic to $h_{2^s,K}$, then $2^{s+1}\times \langle 1\rangle$ is isomorphic to $h_{2^{s+1},K}$, since these two forms are canonically isomorphic to $2\times (2^s\times\langle 1\rangle)$ and $2\times h_{2^s,K}$ respectively. Otherwise, $2^s\times \langle 1\rangle$ is anisotropic. Therefore,  (iii) exactly says that $2^s\times \langle 1\rangle$ represents $-1$.  The multiplicativity property then implies that $2^s\times\langle 1\rangle\simeq -(2^s\times \langle 1\rangle)$. This means that $2^{s+1}\times \langle 1\rangle =0\in W(K),$ which is equivalent to (ii).
\end{proof}

We now use this lemma to rephrase our previous results.

\begin{lem}\label{selfdualexistv2}
Let $L/F$ be a field extension of degree $m$, and let $n$ be an even integer, where $m\geq n$. Write $n=2^{s+1}r,$ where $s\geq 0$ and $r$ is odd.

Assume that there exists a self-dual code $C\subset L^n.$ Then $-1$ is a sum of $2^s$ squares in $L.$
\end{lem}

\begin{thm}\label{nonexistv2}
Let $L/F$ be a field extension of degree $m$, and let $n=2d$ be an even integer, with $m\geq n$. Write $n=2^{s+1}r,$ where $s\geq 0$ and $r$ is odd.

If $-1$ is a sum of $2^s$ squares in $F$, there are no MRD self-dual codes in $L^n$. 
\end{thm}

We now examine the case $n\equiv 2 \ [4]$.

\begin{prop}\label{n24}
Let $n=2d$ be an integer, where $d\geq 1$ is odd, and let $L/F$ be a field extension of degree $m\geq n$. Then the following properties hold:

\begin{enumerate}
    \item Assume that there exists a self-dual code $C\subset L^n$. Then $-1$ is a square in $L.$

\smallskip

    \item If $-1$ is a square in $F$, there are no MRD self-dual codes in $L^n$.

\smallskip

    \item Assume that $L/F$ is a cyclic extension of degree $m=n.$ Then there exists a self-dual MRD code $C\subset L^n$ if and only if $-1$ is a square in $L$ but not in $F$. 
\end{enumerate}
\end{prop}

\begin{proof}
Items $(1)$ and $(2)$ follow from Lemma \ref{selfdualexistv2} and Theorem \ref{nonexistv2} respectively.

Assume now that $L/F$ is a cyclic extension of degree $m=n.$ It remains to construct a self-dual MRD code of dimension $d$ in the case where $-1$ is a square in $L$ but not in $F.$

Since $d$ is odd, we have $\mathbb{Z}/n\mathbb{Z}\simeq \mathbb{Z}/d\mathbb{Z}\times\mathbb{Z}/2\mathbb{Z}$. Elementary Galois theory shows that $L/F$ is the compositum of a cyclic extension $M/F$ of degree $d$ and of a separable quadratic extension $M'/F$. Since $-1$ is not a square in $F$ but is a square in $L$, and since $L/F$ has a unique quadratic subextension, we necessarily have $M'=F(i),$ where $i\in L$ satisfies $i^2=-1.$

Denote by $\tau$ a generator of ${\rm Gal}(M/F).$ Galois theory then shows that there is a unique $F$-automorphism $\sigma$ of $L$ extending $\tau$ and  such that $\sigma(i)=-i.$

Let us introduce some notation. If $K$ is a field and $\xb\in K^t$, for any ring automorphism $\varphi$ of $K$, we denote by $\varphi(\xb)$  the vector of $K^t$ obtained by applying $\varphi$ to its coordinates.

Let  $\cb_0=(c_1,\ldots, c_n)\in L^n$ whose coordinates form an $F$-basis of $L$. Let $C$ be the subspace of $L^n$ generated by $\cb_0, \sigma(\cb_0),\ldots,\sigma^{d-1}(\cb_0)$. This type of code, introduced by Gabidulin in \cite{Gab} for finite fields, is known to be MRD. For sake of completeness, since $F$ is an arbitrary field, we sketch the argument.
By definition, an element of $C$ has the form $\cb=(P(\sigma)(c_1),\ldots,P(\sigma)(c_n))$, where $P\in L[X]$ has degree $\leq d-1.$
Hence, the $F$-span of the coordinates of $\cb$ is the image of $P(\sigma)$, since $(c_1,\ldots,c_n)$ is an $F$-basis of $L$. If $\cb$ is non-zero, $P$ is also non-zero. By \cite[Theorem 5]{Gow}, $\ker(P(\sigma))$ has dimension at most $d-1$, and thus ${\rm wt}(\cb)\geq n-d+1,$ as claimed.

We now prove that there is a suitable choice of $\cb_0$ for which $C$ is self-dual. 

Let $\vb=(v_1,\ldots,v_d)\in M^d$ be a vector whose coordinates form an $F$-basis of $M.$ Then $v_1,\ldots,v_d,i v_1,\ldots,i v_d\in L$ are $F$-linearly independent, and thus form an $F$-basis of $L.$

We then set $\cb_0=(\vb\mid i\vb).$ By definition of $\sigma$, for all $0\leq k\leq d-1$, we have 
$$\sigma^k(\cb_0)=(\tau^k(\vb)\mid (-1)^k i\tau^k(\vb)). $$
For all $0\leq k\leq\ell\leq d-1,$ we then have 
$$\sigma^k(\cb_0)(\sigma^\ell(\cb_0))^t=\sigma^k(\cb_0(\sigma^{\ell-k}(\cb_0))^t)=(1-(-1)^{\ell-k})\sigma^k(\vb(\tau^{\ell-k}(\vb)^t)). $$
If $k=\ell$, this quantity is zero.
Hence, it is enough to find $\vb$ such that $\vb(\tau^j(\vb))^t=0$ for all $1\leq j\leq d-1.$
By \cite[Theorem 5.6 ]{BL}, $M/F$ has a self-dual normal basis ${\vb}=(a,\tau(a),\ldots,\tau^{d-1}(a))$ (with respect to the trace map). For all $1\leq j\leq d-1, $ we have $$\vb(\tau^j(\vb))^t=\sum_{k=0}^{d-1}\tau^k(a)\tau^{j+k}(a)=\sum_{k=0}^{d-1}\tau^k(a\tau^j(a))={\rm Tr}_{M/F}(a\tau^j(a)).$$
Since $1\leq j\leq d-1$, this last quantity is zero by definition of a self-dual normal basis. This concludes the proof.
\end{proof}

\begin{rem}
Since every field extension $L/F$ of degree $2$ is cyclic if $F$ has odd characteristic, this proposition shows that for $n=m=2,$ a self-dual MRD code exists in $L^2$ if and only if $-1$ is a square in $L$ but not in $F$, or equivalently if and only if $-1$ is not a square in $F$ and $L=F(i),$ where $i^2=-1.$
\end{rem}

To end this section, we apply our results to finite fields. 

\begin{thm}\label{finite}
 Let $F=\ff_q$, and let $L=\ff_{q^m}$. Finally, let $n\geq 1$ be an even integer such that $m\geq n.$ Then the following properties hold:

 \begin{enumerate}
     \item If $q$ is even  or $q\equiv 1[ 4]$, there are no self-dual MRD codes in $L^n$.

     \item If  $q$ is odd and $n\equiv 0 \ [4]$, there are no self-dual MRD codes in $L^n$.         
     
     \item If  $q\equiv 3 \ [4]$ and $m=n\equiv 2 \ [4]$, there is a self-dual MRD code in $L^n$.
\end{enumerate}

     In particular, if $m=n$, there is a self-dual MRD code in $L^n$ if and only if $q\equiv 3 \ [4]$ and $n\equiv 2 \ [4].$
\end{thm}

\begin{proof}
If $q$ is even, this comes from Proposition \ref{car2}.

If $q$ is odd, it is known that $-1$ is a square in $\ff_q$ if and only if $q\equiv 1 \ [4]$. Then $(1)$ comes from Proposition \ref{car2} and Theorem \ref{nonexistv2} (since a square is also a sum of $2^s$ squares for all $s\geq 0$).

Assume now that $q$ is odd and $n\equiv 0 \ [4]$. Write $n=2^{s+1}r$, where $r$ is odd, and $s\geq 1$ by assumption. It is a standard fact that $-1$ is the sum of two squares in $\ff_q.$ In particular, it is also a sum of $2^s$ squares in $\ff_q$. Item $(2)$ then comes from Theorem \ref{nonexistv2}. Item $(3)$ is just an application of Proposition \ref{n24}, since the extension $L/F$ is cyclic.

The last part of the theorem is clear.
\end{proof}

\section{Gabidulin self-duality versus Delsarte self-duality}

In this section, we would like to relate our results on Gabidulin self-dual MRD codes $C\subset L^n$  to those already known for Delsarte self-dual MRD codes $\mathscr{C}\subset \M_{m\times n}(F)$.

As already briefly mentioned in the introduction, the two types of duality do not coincide: if $C\subset L^n$ is self-dual, there is no reason for $\mathscr{C}=M_{\ab}(C)\subset \M_{m\times n}(F)$ to be self-dual, as the following example already shows.

\begin{ex}\label{exf3i}
Let $F=\ff_3, $ let $L=\ff_3(i),$ where $i^2=-1,$ and set $C=L(1,i)\subset L^2.$
Then $C$ is clearly self-dual.

Let $\alpha_1=1$ and $\alpha_2=i$, so that $\ab=(\alpha_1,\alpha_2)$ is an $F$-basis of $L$. As an $F$-vector space, $C$ is spanned by $(1,i)$ and $(i,-1)$, so the  code $M_{\ab}(C)$ is the $F$-subspace of $\M_2(F)$ spanned by the identity matrix $\begin{pmatrix}
  1 & 0 \cr 0 & 1  
\end{pmatrix}$ and $\begin{pmatrix}
 0 & -1 \cr 1 & \phantom{-}0   
\end{pmatrix}$.

Clearly, $M_{\ab}(C)$  is not self-dual for the inner product $(M,N)\in \M_2(F)\times \M_2(F)\mapsto {\rm tr}(MN^t)$.

However, another choice of basis $L/F$ yields a self-dual code.

Take this time $\alpha'_1=i$ and $\alpha'_2=1-i$, so $\ab'=(\alpha'_1,\alpha'_2)$ is an $F$-basis of $L$.

We have $(1,i)=(\alpha'_1+\alpha'_2,\alpha'_1)$ and $(i,-1)=(\alpha'_1,-\alpha'_1-\alpha'_2)$, so  $M_{\ab'}(C)$  is spanned by the two matrices $\begin{pmatrix}
    1 & 1 \cr 1 & 0
\end{pmatrix}$ and $\begin{pmatrix}
    1 & -1 \cr 0 & -1
\end{pmatrix}$.

One may check that $M_{\ab'}(C)$  is indeed self-dual.
\end{ex}

The previous example shows that one may obtain a self-dual code $M_{\ab}(C)$ from a self-dual code $C$, provided we choose a suitable $F$-basis of $L$. The next lemma, which is an immediate generalization of \cite[Theorem 21]{Rav}, 
gives a sufficient condition for the existence of such a basis. 

{\bf Notation. }If  $B\in \M_n(F)$ is an invertible symmetric matrix, we will denote by $\varphi_B$ and $\psi_B$ the bilinear forms defined by $$\varphi_B: (\cb,\cb')\in L^n\times L^n\mapsto \cb B(\cb')^t\in L$$
and $$\psi_B:(M,N)\in \M_{m\times n}(F)\times \M_{m\times n}(F)\mapsto {\rm tr }(MBN^t)\in F.$$

In particular, if $B=I_n,$ we get the standard inner products on $L^n$ and $\M_{m\times n}(F)$ respectively.

We then have the following result.

\begin{lem}\label{transfer}
Let $L/F$ be an extension of finite degree $m\geq 1$, where $F$ is a field of arbitrary characteristic. Let $s:L\to F$ be a non-zero $F$-linear map. 

Then any $F$-basis $\ab$ has a dual basis $\ab'$  with respect to $b_s:(x,y)\in L\times L\mapsto s(xy)\in F$, that is, an $F$-basis satisfying $s(\alpha_i\alpha'_j)=\delta_{ij}$ for all $1\leq i,j\leq m.$

Moreover, if $B\in \M_n(F)$ is an invertible symmetric matrix, then for all $C\subset L^n$, we have the equality $$ M_{\ab'}(C^\perp)=M_{\ab}(C)^\perp,$$
where the orthogonal spaces are defined with respect to $\varphi_B$ and $\psi_B$.

In particular, if $b_s$ has an orthonormal basis $\ab$, then for all $C\subset L^n,$
$C$ is self-dual if and only if $M_{\ab}(C)$ is self-dual.
\end{lem}

\begin{proof}
Note that, if $s:L\to F$ is a non-zero $F$-linear map, the bilinear form $b_s: (x,y)\mapsto s(xy)\in F$ is non-degenerate. 

Indeed, if $s(x_0)\neq 0$, then for any $x\in L^\times,$ we have $b_s(x,x^{-1}x_0)=s(x_0)\neq 0.$
In particular, an arbitrary $F$-basis $\ab$ of $L$ has a dual $F$-basis $\ab'$ with respect to $b_s$, that is, an $F$-basis satisfying $s(\alpha_i\alpha'_j)=\delta_{ij}$ for all $1\leq i,j\leq m.$

Let us keep the notation of the lemma, and let $\cb,\cb'\in L^n$.
By Lemma \ref{mc}, we have 

$$\varphi_B(\cb,\cb')=\ab M_{\ab}(\cb)BM_{\ab'}(\cb')^t (\ab')^t.$$

If we write $ M_{\ab}(\cb)BM_{\ab'}(\cb')^t=(\lambda_{ij})\in \M_m(F)$, we then get 
$$\varphi_B(\cb,\cb')=\sum_{i,j}\lambda_{ij}\alpha_i\alpha'_j.$$

Since $\ab'$ is the dual basis of $\ab$ with respect to $b_s$, applying $s$ on both sides of the equality above then yields $$s(\varphi_B(\cb,\cb'))=\sum_{i=1}^m\lambda_{ii} ={\rm tr}(M_{\ab}(\cb)B(M_{\ab'}(\cb')^t)=\psi_B(M_{\ab}(\cb),M_{\ab'}(\cb')).$$

It follows that for all $C\subset L^n$, we have $M_{\ab'}(C^\perp)\subset M_{\ab}(C)^\perp,$ and equality follows by comparing dimensions. 

The last part of the lemma is then clear, since a self-dual basis is nothing but an orthonormal basis.
\end{proof}

Given a field extension, the existence of  a linear map $s$ such that $b_s$ has an orthonormal basis is a difficult question. When $L/F$ is a separable extension, every linear map $s:L\to F$ has the form $t_\lambda: x\in L\mapsto {\rm Tr}_{L/F}(\lambda x)\in F$ for a suitable $\lambda\in L.$

Indeed,  in this case the trace form  $$\mathcal{T}_{L/F}: (x,y)\in L\times L\mapsto {\rm Tr}_{L/F}(xy)\in F$$ is non-degenerate, and therefore induces an isomorphism from $L$ onto the space of $F$-linear forms on $L$,  which maps $\lambda$ onto $t_\lambda$.

The question may then be rephrased as follows. If $L/F$ is a finite  separable field extension, does there exist $\lambda\in L^\times$ such that the $F$-bilinear form 
$$\mathcal{T}_{L/F,\lambda}: (x,y)\in L\times L\mapsto {\rm Tr}_{L/F}(\lambda xy)\in F $$ has an orthonormal basis ?

Even reformulated this way, the answer is not completely known for a field as simple as $F=\mathbb{Q}$. A result of Bender \cite{Ben} says that the answer is positive for any field extension $L/F$ of odd degree, where $F$ is a number field, but this result is not really useful in the context of MRD self-dual codes in view of Theorem \ref{nonexistv2}.

To point out the fact that the answer heavily depends on $L/F$, let us mention that if $F=\mathbb{Q}$ and $L=\mathbb{Q}(\sqrt{d})$ is a quadratic extension, then it is easy to show that such a $\lambda$ exists if and only if $d$ is a sum of two squares in $\mathbb{Q}.$

In the case of finite fields, everything is simpler, since such a $\lambda$ always exists. In other words, we have the following proposition.

\begin{prop}\label{twoduals}
Let $F=\mathbb{F}_q$, where $q$ is odd, and let $L/F$ be a finite extension of degree $m\geq 1.$

Let $\delta\in F^\times$ be a representative of $\det(\mathcal{T}_{L/F})\in F^\times/F^{\times 2},$ where $$\mathcal{T}_{L/F}: (x,y)\in L\times L\mapsto {\rm Tr}_{L/F}(xy)\in F, $$ and let $\lambda\in L^\times$ such that $N_{L/F}(\lambda)=\delta.$
Then the bilinear form $$\mathcal{T}_{L/F,\lambda}: (x,y)\in L\times L\mapsto {\rm Tr}_{L/F}(\lambda xy)\in F $$
 has an orthonormal basis $\ab$.
 
In particular, for such a basis $\ab$, and  for every  code $C\subset L^n$, the code 
$C$ is self-dual with respect to $\varphi_B$ if and only if $M_{\ab}(C)\subset \M_{m\times n}(F)$ is self-dual with respect to $\psi_B$.
\end{prop}

\begin{proof}
Let us keep the notation of the proposition.
By \cite[Chapter 2, Theorem 5.12]{Sch},  we have $\det(\mathcal{T}_{L/F,\lambda})=N_{L/F}(\lambda)\delta\in F^\times/F^{\times 2}.$

Since $F$ is a finite field, the norm map is surjective, so we may choose $\lambda\in L^\times$ such that  $N_{L/F}(\lambda)=\delta.$ Hence, we have $\det(\mathcal{T}_{L/F,\lambda})=1\in F^\times/F^{\times 2}.$

Therefore, the bilinear forms $\mathcal{T}_{L/F,\lambda}$ and $m\times \langle 1\rangle$ have same rank and determinant. By \cite[Chapter 2, Theorem 3.8]{Sch}, they are isomorphic. But this exactly means that $\mathcal{T}_{L/F,\lambda}$ has an orthonormal basis.
The rest of the proposition is a direct application of Lemma \ref{transfer}.
\end{proof}

\begin{ex}
Let $F=\ff_3$ and $L=\ff_3(i),$ where $i^2=-1.$ The determinant of $\mathcal{T}_{L/F}$ is easily seen to be $-1\in F^\times/F^{\times 2}$. Then $\lambda=1+i$ has norm $-1$, and $\mathcal{T}_{L/F,\lambda}$ has rank two and trivial determinant, so it should have an orthonormal basis, according to Proposition \ref{twoduals}. One may check that $\ab'=(i,1-i)$ is indeed such an orthonormal basis, which explains more conceptually the second part of Example \ref{exf3i}.
\end{ex}

Proposition \ref{twoduals} allows us to put out our results in perspective with the results of \cite{Neb}.

Assume that $F=\mathbb{F}_q$, so that $L=\mathbb{F}_{q^m}$. Fixing an  $F$-basis $\ab$ of $L$ as in Proposition \ref{twoduals}, Theorem \ref{finite} then  translates as follows:

Let $q\geq 2$ be a prime power, and let $m,n\geq 1$ be two integers, with $m\geq n.$

\begin{enumerate}
 
 \item If $q$ is even  or $q\equiv 1[ 4]$, there is no code $C\subset L^n$ such that $M_{\ab}(C)$ is a self-dual MRD code.

     \item If  $q$ is odd and $n\equiv 0 \ [4]$, there is no code    $C\subset L^n$ such that $M_{\ab}(C)$ is a self-dual MRD code.   
     
     \item If  $q\equiv 3 \ [4]$ and $m=n\equiv 2 \ [4]$, there is a self-dual MRD code $C\subset L^n$. In this case, $M_{\ab}(C)$ is a self-dual MRD code.  
\end{enumerate}

     In particular, if $m=n$, there is a self-dual MRD code of the form $M_{\ab}(C),$ where $C\subset L^n$, if and only if $q\equiv 3 \ [4]$ and $n\equiv 2 \ [4].$

     Items $(1)$ and $(2)$ seem to give strong evidence that self-dual MRD codes $\Cc\subset \M_{m\times n}(F)$ should not exist if $n\equiv 0 \ [4]$ or $-1$ is square in $\mathbb{F}_q$, as (mildly) suggested by Nebe and Willems in \cite{Neb} when $m=n$.

     Note that Items $(1)$ and $(2)$ do not assume $m=n$. However, the assumption $m=n$ is nevertheless crucial for the previous conjecture to have a chance to be true. Indeed, the code $\mathscr{C}\subset\M_{4\times 2}(\mathbb{F}_5)$ generated by the four matrices

     $$\begin{pmatrix}
        4 & 0\cr 4 & 1 \cr  0 & 3 \cr 2 & 2
     \end{pmatrix}, \ \begin{pmatrix} 4 & 0\cr  1 & 1 \cr  0 & 4 \cr 0 & 1\end{pmatrix}, \ \begin{pmatrix}
          2 & 2\cr  0 & 2 \cr  1 & 4 \cr 0 & 4
     \end{pmatrix}, \  \begin{pmatrix}
         3  & 1\cr 0  &1 \cr 3 & 3 \cr 4 & 0
     \end{pmatrix} $$

     is a self-dual MRD code (this is the code of $\M_ {4\times 2}(\mathbb{F}_5)$ corresponding to the generating matrix $G_8$ of Table $5$ in Morrison's paper \cite{Mor}, after some reformulation in order to agree with our notation).

     This apparent contradiction with Item (1) may be resolved by noticing that not all codes $\Cc\subset \M_{m\times n}(F)$ may be written under the form $M_{\ab}(C)$ where $C\subset L^n$.

     Indeed, the code $M_{\ab}(C)$ inherits from $C$ a structure of an $L$-vector space, given by $$(\lambda, M_{\ab}(\cb))\in L\times M_{\ab}(C)\mapsto  M_{\ab}(\lambda \cb)\in  M_{\ab}(C).$$

     This example just says that there are self-dual MRD codes $\Cc\subset \M_{4\times 2}(\mathbb{F}_5)$ which are not obtained from a self-dual MRD code $C\subset L^2,$ where $L=\mathbb{F}_{5^4}.$ 

     Nevertheless, one could still ask the following question:

     {\bf Question. }Let $F=\mathbb{F}_q$, with $q$ odd. Assume that $m=n$ is even, and that either $-1$ is a square in $F$ or that $n\equiv 0 \ [4]$. Is it true that there is no self-dual MRD code in $\M_n(F)$ ?

\section{The case of characteristic two revisited}

Let $F$ be a field of characteristic two, and let $L/F$ be a field extension of degree $m\geq 1,$ with $m\geq n$.

Proposition \ref{car2} shows that there is no self-dual MRD codes $C\subset L^n$ in this case. Even if Lagrangian MRD codes do not exist in odd characteristic (Theorem \ref{MRDhyp}), we now proceed to show that they may exist in characteristic two, providing a natural substitute for classical self-dual MRD codes in this case. 

\begin{prop}\label{proplag2}
Let $F$ be a field of characteristic two.

Let $n=2d,$ where $d$ is odd or $d\equiv 2 \ [4],$ and let $L/F$ be an extension of degree $m=n$ with Galois group $\mathbb{Z}/d\mathbb{Z}\times \mathbb{Z}/2\mathbb{Z}$. Then there exists a Lagrangian MRD code $C\subset L^n$.  
\end{prop}

\begin{proof}
The proof being almost identical to the proof of   
Proposition \ref{n24}, we only point out the necessary 
modifications.

Elementary Galois theory shows that $L/F$ is the compositum of a cyclic extension $M/F$ and of a separable quadratic extension $F(\alpha)/F$, where $\alpha^2-\alpha\in F.$

Denote by $\tau$ a generator of ${\rm Gal}(M/F).$
By \cite[Theorem 6.1]{BL}, the assumption on $d$ implies that $M/F$ has a self-dual normal basis ${\vb}=(a,\tau(a),\ldots,\tau^{d-1}(a)).$

Moreover, Galois theory shows that there is a unique $F$-automorphism $\sigma$ of $L$ extending $\tau$ and  such that $\sigma(\alpha)=\alpha+1.$

Let  $\cb_0=(\vb\mid \alpha\vb)\in L^n$. Then its coordinates form an $F$-basis of $L$. Let $C$ be the subspace of $L^n$ generated by $\cb_0, \sigma(\cb_0),\ldots,\sigma^{d-1}(\cb_0)$, so that $C$ is an MRD code.

For all $0\leq k\leq d-1$, we have $\sigma^k(\cb_0)=(\tau^k(\vb)\mid (\alpha+k)\tau^k(\vb)).$

Moreover, for all $0\leq k\leq \ell\leq d-1$, we have this time $$h_{n,L}(\tau^k(\cb_0),\tau^\ell(\cb_0))=(\ell-k)\tau^k(\vb (\tau^{\ell-k}(\vb))^t).$$
We may now conclude as in the proof of Proposition \ref{n24}.
\end{proof}

\begin{rem}
If $F$ is a finite field, then the Galois group of $L/F$ is necessarily cyclic. Hence, the previous result can only be applied in the case where $d$ is odd, that is when $n\equiv 2 \ [4]$.
\end{rem}

We now address the question of the relationship  between Lagrangian MRD codes in $L^n$ and MRD codes $\Cc\subset \M_{m\times n}(F)$ which are self-dual for a suitable notion of duality. As we have seen before, we cannot take the standard inner product on $\M_{m\times n}(F)$, so we have to replace it by a suitable non-degenerate symmetric bilinear form.

Note that the bilinear form $\varphi_{H_n}$ on $L^n$ is just the standard hyperbolic form $h_{n,L}.$
In view of Lemma \ref{transfer}, a natural choice would be 
 the bilinear form $$\psi_{H_n}:(M,N)\in\M_{m\times n}(F)\times \M_{m\times n}(F)\mapsto {\rm tr}(MH_nN^t)\in F,$$ where $H_n=\begin{pmatrix}0 & I_d \cr I_d & 0\end{pmatrix}\in\M_n(F)$.

 \begin{defn}
 A code $\Cc\subset \M_{m\times n}(F)$ is called a {\it Lagrangian code}  if it self-dual with respect to $\psi_{H_n}$. 
\end{defn}

Lemma \ref{transfer} will give us the required duality statement, provided that we can find a non-zero linear map $s:L\to F$ such that the corresponding bilinear form $b_s:L\times L\to F$ has an orthonormal basis. 

Contrary to the case of odd characteristic, the answer is completely known for $F$-algebras of the form $F[X]/(f)$ (see \cite{Be}).
In particular, if $L/F$ is a separable field extension (which is more than enough in the context of rank-metric codes), one may take $s={\rm Tr}_{L/F}$ (see \cite[Lemma 3.2]{Be}).

Therefore, we get the following corollary.

\begin{cor}
 Let $F$ be a field of characteristic two, and let $L/F$ be a finite separable extension of degree $m\geq 1.$ Let $\ab=(\alpha_1,\ldots,\alpha_m)$ be an $F$-orthonormal basis with respect to the bilinear form $\mathcal{T}_{L/F}:(x,y)\in L\times L\mapsto {\rm Tr}_{L/F}(xy).
 $
 Then for any $C\subset L^n$, we have the equality $$ M_{\ab}(C^\perp)=M_{\ab}(C)^\perp.$$

In particular, $C$ is a Lagrangian code if and only if $M_{\ab}(C)$ is a Lagrangian code.
\end{cor}

Note that Proposition \ref{proplag2} and the corollary above yield the existence of Lagrangian MRD codes $\mathscr{C}\subset \M_n(\mathbb{F}_q)$ when $q$ is even and $n\equiv 2 \ [4]$.

\section{Conclusion}   
     
In this paper, we have studied the question of self-dual MRD codes à la Gabidulin, and proved necessary conditions on the field extension  to have the existence of such codes, as well as sufficient conditions on the base field to have their non-existence. Along the way, we proved the non-existence of MRD codes which are self-dual for the standard hyperbolic form.
We also completely solved the question of the existence of self-dual MRD codes over finite fields when $m=n$. We also compared the two notions of duality and put our results in perspective with the existing results for Delsarte self-dual MRD codes. We have shown strong evidence for the fact that no self-dual MRD code in $\M_n(\mathbb{F}_q)$ should exist if $-1$ is a square in $\mathbb{F}_q$ or $n\equiv 0 \ [4].$
Finally, we proved that there is no MRD self-dual MRD codes in characteristic two for the standard inner product on $L^n$, and we proposed a notion of duality on $L^n$ and on $\M_{m\times n}(F)$ for which self-dual MRD codes may exist.

{\bf Acknowledgments. }The author thanks warmly the anonymous referees for their insightful comments and suggestions, which greatly improved the exposition of this paper and drastically simplified the arguments in the proof of Theorem \ref{MRDhyp}.

\end{document}